\documentclass[12pt,a4paper]{article}
\usepackage{PRIMEarxiv}
\usepackage[T1]{fontenc}
\usepackage[utf8]{inputenc}
\usepackage{lmodern}
\usepackage{microtype}
\usepackage{graphicx}
\usepackage{geometry}
\usepackage{amsmath,amssymb,amsthm}
\usepackage{booktabs}
\usepackage{tikz}
\usepackage{pgfplots}
\usepackage{xcolor}
\usepackage{float}
\floatstyle{ruled}
\newfloat{algorithm}{tbp}{loa}
\floatname{algorithm}{Algorithm}
\usepackage{hyperref}
\usepackage{needspace}

\usetikzlibrary{arrows.meta,positioning}
\pgfplotsset{compat=1.18}

\definecolor{RTwoBlue}{RGB}{40,110,190}
\definecolor{RTwoPurple}{RGB}{120,80,170}
\definecolor{RTwoOrange}{RGB}{220,130,35}
\definecolor{RTwoGreen}{RGB}{35,145,95}
\definecolor{RTwoGray}{RGB}{100,105,115}

\hypersetup{
  colorlinks=true,
  linkcolor=blue!55!black,
  citecolor=blue!55!black,
  urlcolor=blue!55!black
}

\newcommand{\Rtwo}{R_2}
\newcommand{\argmin}{\operatorname*{arg\,min}}
\newtheorem{lemma}{Lemma}
\newtheorem{theorem}{Theorem}

\title{Exact and Fast Subset Selection Algorithms for the Bi-objective Integral \texorpdfstring{$R_2$}{R2} Indicator}
\author{Michael T. M. Emmerich\\Faculty of Information Technology, University of Jyv\"askyl\"a, Finland\\\href{https://orcid.org/0000-0002-7342-2090}{ORCID: 0000-0002-7342-2090}}
\date{\today}
\AtBeginDocument{\setlength{\headheight}{15pt}}

\begin{document}
\maketitle

\begin{abstract}
We study fixed-cardinality subset selection for the exact integral bi-objective $R_2$ indicator with a uniform continuum of weighted Tchebycheff scalarizing functions. The indicator measures the area under the lower envelope of scalarizing losses over weight space, rather than a finite sample average over weight vectors. For a sorted bi-objective Pareto-front approximation, represented by points ordered by increasing first objective and decreasing second objective, we derive an exact adjacent-neighbor decomposition of this integral objective into boundary terms, unary diagonal corrections, and selected-neighbor transition terms. This yields an exact Bellman dynamic program with $O(kn^2)$ running time for selecting $k$ of $n$ candidate points. We then prove that the transition matrix is Monge. This gives a divide-and-conquer implementation with $O(kn\log n)$ running time and, more strongly, a staircase matrix-search implementation with $O(kn)$ running time under constant-time arithmetic comparisons. The matrix-search proof is presented through a lower-envelope sweep over single-crossing transition functions and includes the triangular feasibility condition $i<j$. The algorithms are exact for the continuous integral $R_2$ setting and are distinct from finite-weight-vector approximations, although they are related to earlier exact and dynamic-programming work on two-dimensional indicator-based subset selection, including hypervolume subset selection. Reproducible Python code compares exhaustive enumeration, the direct left-to-right dynamic program, the divide-and-conquer dynamic program, and the matrix-search implementation under explicit consistency checks.
\end{abstract}
\keywords{multiobjective optimization; biobjective optimization; integral $\Rtwo$ indicator; subset selection; Pareto-front approximation; dynamic programming; Monge arrays; matrix search; Tchebycheff scalarization}

\section{Introduction}

Quality indicators are often used to compare or select finite nondominated approximations of a Pareto front. In two-objective minimization, the candidate archive can be sorted by increasing first objective and decreasing second objective. The increasing--decreasing order is introduced only as an algorithmic representation of the nondominated archive. The algorithmic question considered in this paper is the following: given such an archive and a cardinality budget $k$, which $k$ points should be retained if quality is measured by the \textit{exact integral $R_2$ indicator}?

The classical $R_2$ idea evaluates a set through scalarizing utility or loss functions \cite{hansen1998evaluating}. For a weight $\lambda\in[0,1]$ and a utopian point $z^+$, the bi-objective weighted Tchebycheff loss of a point $a$ is
\[
  g_\lambda(a;z^+)=\max\{\lambda(a_1-z_1^+),(1-\lambda)(a_2-z_2^+)\}. 
\]
A set is evaluated via its best point for each weight. The standard finite-weight implementation averages these best losses over a finite set of weights. The \textit{Integral $R_2$} indicator considered in our work \cite{schaepermeier2024reinvestigating, schaepermeier2025r2v2, jaszkiewicz2025exact} yields a Pareto compliant (rather than weakly Pareto compliant) version of the classical $R_2$ indicator by integrating over the entire weight interval instead of a finite set of weight vectors. After translating $z^+$ to the origin, the relevant envelope is
\[
  r_{\mathcal P}(\lambda)=\min_{a\in\mathcal P} g_\lambda(a;0),
  \qquad
  R_2(\mathcal P)=\int_0^1 r_{\mathcal P}(\lambda)\,d\lambda .
\]
Here the integration variable $\lambda$ moves continuously from full emphasis on the second objective at $\lambda=0$ to full emphasis on the first objective at $\lambda=1$. For each value of $\lambda$, the archive is first reduced to the point with the smallest weighted Tchebycheff loss; only then is this best loss integrated over all weights. Thus $R_2(\mathcal P)$ is the area under the lower envelope of the pointwise losses over the whole preference interval, with respect to the uniform measure $d\lambda$. This is different from numerical quadrature or a finite grid of weights: every weight in $[0,1]$ contributes, and a point contributes only on those subintervals on which its loss curve forms the lower envelope. Since the formulation uses losses, smaller integral area means a better approximation.

\Needspace{6\baselineskip}
Figure~\ref{fig:tchebycheff-shadows} shows the finite-archive viewpoint used in this paper. The Tchebycheff-shadow terminology and visualization are adapted from Emmerich~\cite{tchebEmmerich}. Each point $p_i=(x_i,y_i)$ of a finite Pareto-front approximation casts a Tchebycheff shadow over the weight interval,
\[
  \tau_i(\lambda)=\max\{\lambda x_i,(1-\lambda)y_i\}.
\]
The set value is obtained by taking the lower envelope of these shadows,
\[
  r_{\mathcal P}(\lambda)=\min_i \tau_i(\lambda),
\]
and integrating the shaded area under this envelope. Subset selection can therefore be viewed as retaining only $k$ shadows while preserving the integral envelope as well as possible. This finite-envelope viewpoint is the starting point of the exact bi-objective computation of Sch\"apermeier and Kerschke \cite{schaepermeier2024reinvestigating,schaepermeier2025r2v2}. The Tchebycheff-shadow geometric explanation used here follows the terminology and interpretation in Emmerich~\cite{tchebEmmerich}. Appendix~\ref{app:r2-integral-derivation} gives a detailed expository derivation of the elementary integral calculation behind the adjacent-neighbor formula used below. That appendix is included only to make the paper more self-contained; it adds no new result beyond the integral $R_2$ analysis of Sch\"apermeier and Kerschke and the shadow interpretation just cited. Closely related integral and exact-calculation questions for the $R_2$ indicator, including higher-dimensional settings, are also studied independently by Jaszkiewicz and Zielniewicz \cite{jaszkiewicz2025exact}. The present paper uses the bi-objective integral setting for fixed-cardinality subset selection.

\begin{figure}[t]
\centering
\begin{tikzpicture}
\begin{scope}[xshift=0cm]
\begin{axis}[
  width=0.46\textwidth, height=0.34\textwidth, axis lines=left,
  xmin=0, xmax=0.72, ymin=0, ymax=0.68,
  xlabel={$f_1$}, ylabel={$f_2$}, grid=both,
  grid style={line width=.1pt, draw=gray!18}, major grid style={line width=.2pt, draw=gray!32},
  tick label style={font=\scriptsize}, label style={font=\small}, clip=false]
  \addplot[RTwoGray, thick, densely dashed] coordinates {(0.10,0.60) (0.18,0.43) (0.30,0.30) (0.45,0.18) (0.62,0.10)};
  \draw[RTwoPurple!55, thick] (axis cs:0,0.60) -- (axis cs:0.10,0.60) -- (axis cs:0.10,0);
  \draw[RTwoBlue!55, thick]   (axis cs:0,0.43) -- (axis cs:0.18,0.43) -- (axis cs:0.18,0);
  \draw[RTwoOrange!75, thick] (axis cs:0,0.30) -- (axis cs:0.30,0.30) -- (axis cs:0.30,0);
  \draw[RTwoGreen!65!black, thick] (axis cs:0,0.18) -- (axis cs:0.45,0.18) -- (axis cs:0.45,0);
  \draw[RTwoGray, thick] (axis cs:0,0.10) -- (axis cs:0.62,0.10) -- (axis cs:0.62,0);
  \addplot[only marks, mark=*, mark size=2.1pt, RTwoPurple] coordinates {(0.10,0.60)};
  \addplot[only marks, mark=*, mark size=2.1pt, RTwoBlue] coordinates {(0.18,0.43)};
  \addplot[only marks, mark=*, mark size=2.1pt, RTwoOrange] coordinates {(0.30,0.30)};
  \addplot[only marks, mark=*, mark size=2.1pt, RTwoGreen!70!black] coordinates {(0.45,0.18)};
  \addplot[only marks, mark=*, mark size=2.1pt, RTwoGray] coordinates {(0.62,0.10)};
  \node[font=\scriptsize, fill=white, inner sep=0.8pt, anchor=south west, text=RTwoPurple] at (axis cs:0.105,0.602) {$p_1$};
  \node[font=\scriptsize, fill=white, inner sep=0.8pt, anchor=south west, text=RTwoBlue] at (axis cs:0.185,0.432) {$p_2$};
  \node[font=\scriptsize, fill=white, inner sep=0.8pt, anchor=south west, text=RTwoOrange] at (axis cs:0.305,0.302) {$p_3$};
  \node[font=\scriptsize, fill=white, inner sep=0.8pt, anchor=south west, text=RTwoGreen!70!black] at (axis cs:0.455,0.182) {$p_4$};
  \node[font=\scriptsize, fill=white, inner sep=0.8pt, anchor=south west, text=RTwoGray] at (axis cs:0.625,0.102) {$p_5$};
  \node[font=\scriptsize, fill=white, inner sep=1.2pt, anchor=west] at (axis cs:0.35,0.58) {finite Pareto-front};
  \node[font=\scriptsize, fill=white, inner sep=1.2pt, anchor=west] at (axis cs:0.36,0.53) {approximation};
  \node[font=\scriptsize, anchor=north east] at (axis cs:0,0) {$z^+$};
\end{axis}
\end{scope}
\begin{scope}[xshift=0.51\textwidth]
\begin{axis}[
  width=0.46\textwidth, height=0.34\textwidth, axis lines=left,
  xmin=0, xmax=1, ymin=0, ymax=0.32,
  xlabel={$\lambda$ in $(\lambda,1-\lambda)$}, ylabel={Tchebycheff loss},
  xtick={0,0.2,0.4,0.6,0.8,1.0}, ytick={0,0.08,0.16,0.24,0.32},
  grid=both, grid style={line width=.1pt, draw=gray!18}, major grid style={line width=.2pt, draw=gray!32},
  tick label style={font=\scriptsize}, label style={font=\small}, clip=true]
  \addplot[draw=none, fill=RTwoBlue!18, domain=0:1, samples=300]
    {min(min(min(min(max(0.10*x,0.60*(1-x)),max(0.18*x,0.43*(1-x))),max(0.30*x,0.30*(1-x))),max(0.45*x,0.18*(1-x))),max(0.62*x,0.10*(1-x)))} \closedcycle;
  \addplot[RTwoPurple!70, densely dashed, domain=0:1, samples=100] {max(0.10*x,0.60*(1-x))};
  \addplot[RTwoBlue!70, densely dashed, domain=0:1, samples=100] {max(0.18*x,0.43*(1-x))};
  \addplot[RTwoOrange!85, densely dashed, domain=0:1, samples=100] {max(0.30*x,0.30*(1-x))};
  \addplot[RTwoGreen!70!black, densely dashed, domain=0:1, samples=100] {max(0.45*x,0.18*(1-x))};
  \addplot[RTwoGray, densely dashed, domain=0:1, samples=100] {max(0.62*x,0.10*(1-x))};
  \node[font=\scriptsize, fill=white, inner sep=0.8pt, text=RTwoPurple] at (axis cs:0.52,0.288) {$\tau_1$};
  \node[font=\scriptsize, fill=white, inner sep=0.8pt, text=RTwoBlue] at (axis cs:0.29,0.305) {$\tau_2$};
  \node[font=\scriptsize, fill=white, inner sep=0.8pt, text=RTwoOrange] at (axis cs:0.88,0.270) {$\tau_3$};
  \node[font=\scriptsize, fill=white, inner sep=0.8pt, text=RTwoGreen!70!black] at (axis cs:0.86,0.235) {$\tau_4$};
  \node[font=\scriptsize, fill=white, inner sep=0.8pt, text=RTwoGray] at (axis cs:0.73,0.245) {$\tau_5$};
  \addplot[RTwoBlue, very thick, domain=0:1, samples=300]
    {min(min(min(min(max(0.10*x,0.60*(1-x)),max(0.18*x,0.43*(1-x))),max(0.30*x,0.30*(1-x))),max(0.45*x,0.18*(1-x))),max(0.62*x,0.10*(1-x)))};
  \node[font=\scriptsize, fill=white, inner sep=1.2pt, anchor=west] at (axis cs:0.48,0.265) {individual shadows $\tau_i$};
  \node[font=\scriptsize, fill=white, inner sep=1.2pt, anchor=west] at (axis cs:0.42,0.040) {area = integral $R_2$};
\end{axis}
\end{scope}
\end{tikzpicture}
\caption{Tchebycheff shadows of a finite Pareto-front approximation. Each archive point defines one shadow function over the weight interval; the exact integral $R_2$ value is the shaded area under the lower envelope of these shadows.}
\label{fig:tchebycheff-shadows}
\end{figure}
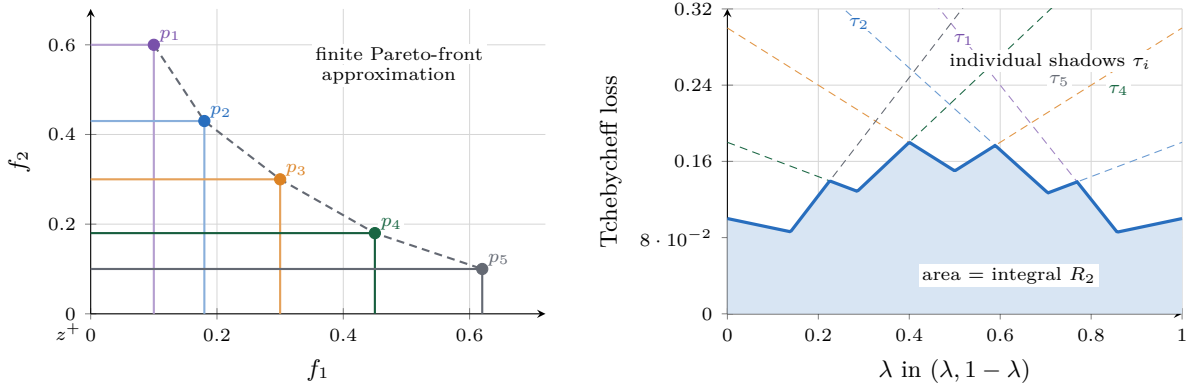

The subset-selection version of this problem is important for archiving, benchmarking, and postprocessing, and selection in set-oriented and evolutionary search. Archives produced by multiobjective optimizers may contain many nondominated points, while visual inspection, decision-support sessions, or downstream expensive evaluation may require a small representative subset. A finite-weight approximation would turn the problem into a discrete aggregation problem whose outcome depends on the chosen grid of weights. The exact integral criterion removes this discretization parameter in the bi-objective setting and asks for the best cardinality-$k$ subset with respect to the continuous scalarization envelope.

The contributions of this paper are as follows.
\begin{itemize}
\item We derive an exact adjacent-neighbor decomposition of the bi-objective integral $R_2$ value of a sorted selected subset.
\item We obtain a Bellman dynamic program for exact fixed-cardinality subset selection in $O(kn^2)$ time.
\item We prove a Monge property of the transition matrix and use it to obtain a divide-and-conquer dynamic program with $O(kn\log n)$ running time.
\item We sharpen the layer computation to a staircase matrix-search/lower-envelope algorithm with $O(n)$ time per dynamic-programming layer, hence $O(kn)$ total time.
\item We provide reproducible computational experiments comparing exhaustive enumeration, direct dynamic programming, divide-and-conquer dynamic programming, and the matrix-search implementation under explicit consistency checks.
\end{itemize}

The reproducibility material is available from the repository listed in the Code availability section. 

The paper is organized as follows. Section~\ref{sec:related-work} places the result in the context of integral $R_2$ indicators, two-dimensional indicator subset selection, and Monge matrix search. Section~\ref{sec:preliminaries} introduces the exact integral subset-selection problem and derives the adjacent-neighbor decomposition. Section~\ref{sec:algorithms} presents the three exact algorithms in a parallel format: the direct Bellman dynamic program, the divide-and-conquer dynamic program, and the matrix-search dynamic program, each with motivation, pseudocode, correctness, and complexity discussion. Section~\ref{sec:computational-study} reports empirical verification tests and CPU-time experiments. Section~\ref{sec:conclusions} summarizes the results and gives an outlook. Appendix~\ref{app:r2-integral-derivation} gives an expository derivation of the exact bi-objective integral $R_2$ decomposition used in the main text; it adds no new result, but records the elementary calculation behind the formula. Appendix~\ref{app:matrix-search-intro} gives a gentle introduction to Monge matrix search and the staircase lower-envelope viewpoint.

\section{Related work}
\label{sec:related-work}

The $R_2$ indicator originates in the framework of Hansen and Jaszkiewicz for evaluating approximations to the nondominated set under a distribution of utility functions \cite{hansen1998evaluating}. In practice, $R_2$ has often been used through finite weight-vector approximations and in indicator-based search methods, for example, in the $R_2$ indicator-based multiobjective search framework of Brockhoff, Wagner, and Trautmann \cite{brockhoff2015r2search}. Sch\"apermeier and Kerschke recently revisited the indicator from the continuous/integral perspective and showed that, for a uniform continuum of Tchebycheff utility functions, the resulting bi-objective indicator can be computed exactly and is Pareto-compliant \cite{schaepermeier2024reinvestigating,schaepermeier2025r2v2}. Jaszkiewicz and Zielniewicz independently study exact calculation and properties of an integral $R_2$ multiobjective quality indicator, including higher-dimensional cases \cite{jaszkiewicz2025exact}. Together, these works motivate treating integral $R_2$ as an exact continuous indicator rather than as a finite weight-vector average. The present paper uses the bi-objective integral setting as its starting point and studies exact fixed-cardinality subset selection for a sorted Pareto-front approximation.

Dynamic programming and exact algorithms for two-dimensional indicator-based subset selection are not new in themselves. For the hypervolume indicator, Kuhn, Fonseca, Paquete, Ruzika, Duarte, and Figueira studied two-dimensional hypervolume subset selection, including formulations and algorithms for that special geometric setting \cite{kuhn2016hypervolume}. Bringmann, Friedrich, and Klitzke also studied two-dimensional subset selection for the hypervolume and epsilon indicators and gave efficient algorithms for those indicators \cite{bringmann2014twodimensional}. More recently, Korogi and Tanabe proposed polynomial-time dynamic programs for the bi-objective indicator-based subset selection problem with IGD, IGD$^+$, $R_2$, and NR2  under practical assumptions \cite{korogi2025dynamic}. Their conference paper formulates $R_2$ and NR2 in terms of a finite weight-vector set. The contribution of the present paper should therefore be read narrowly: it provides a self-contained derivation, implementation, and Monge acceleration for exact bi-objective integral $R_2$ subset selection, i.e., for the specific Pareto compliant integral version of the $R_2$ indicator formulated in \cite{schaepermeier2024reinvestigating,schaepermeier2025r2v2} and \cite{jaszkiewicz2025exact}.

The structural motivation also parallels the Monge/submodular-lattice approach used in one-dimensional Riesz-energy subset selection \cite{emmerich2026riesz}. In the Riesz-energy case, all selected pairs interact, so an ordinary left-to-right Bellman recurrence is not exact; in the integral $R_2$ case, the objective has only adjacent selected-neighbor terms, and the Bellman recurrence is exact.

Monge arrays and totally monotone matrices
are a classical source of faster dynamic programs. The SMAWK algorithm of Aggarwal, Klawe, Moran, Shor, and Wilber computes row or column minima of a totally monotone matrix in linear time under implicit matrix access \cite{aggarwal1987smawk}. The matrix-search result in Section~\ref{sec:matrix-search} uses the same principle but specializes it to the staircase-shaped feasible domain of the recurrence and presents the algorithm as a lower-envelope sweep over single-crossing transition functions.

\section{Preliminaries and problem setting}
\label{sec:preliminaries}
\label{sec:problem-exposition}

We consider a minimization archive of nondominated points
\[
  p_i=(x_i,y_i),\qquad i=1,\ldots,n,
\]
measured relative to a utopian point that has been translated to the origin. The archive is sorted as
\[
  0<x_1<\cdots <x_n,\qquad y_1>\cdots >y_n>0.
\]

Next, we will introduce the analytic expression of the exact integral $R_2$. For a detailed geometrical motivation, we refer to the original papers. For a selected subset $S\subseteq\{1,\ldots,n\}$, the exact integral $R_2$ value for the uniform bi-objective weighted Tchebycheff family is
\begin{equation}
  \Rtwo(S)
  =\int_0^1 \min_{i\in S}\max\{w x_i,(1-w)y_i\}\,dw.
  \label{eq:r2-integral}
\end{equation}
This is the continuous-utility version studied in the bi-objective case by Sch\"apermeier and Kerschke \cite{schaepermeier2024reinvestigating,schaepermeier2025r2v2} and, from a broader exact-calculation perspective including higher dimensions, by Jaszkiewicz and Zielniewicz \cite{jaszkiewicz2025exact}. The present paper focuses on fixed-cardinality subset selection for the bi-objective exact integral objective, rather than on a finite weight-vector approximation.

The subset-selection problem treated here is:
\[
  \text{minimize } \Rtwo(S) \quad\text{subject to}\quad |S|=k,
\]
where $k$ is fixed in advance. Lower values are better. The input order is important: we restrict to one sorted Pareto-front approximation, not to arbitrary unordered two-dimensional point sets.

\subsection{Exact adjacent-neighbor decomposition}
\label{sec:adjacent-decomposition}

Let
\[
  S=(i_1<i_2<\cdots<i_k)
\]
be the selected index vector. Define the corner interaction
\begin{equation}
  A_{ij}=\frac{x_j y_i}{2(x_j+y_i)}.
  \label{eq:transition}
\end{equation}
The diagonal value $A_{ii}$ is the triangular correction associated with the kink of the single point $p_i$.

For two consecutive selected points $i<j$, the relevant switch of the lower Tchebycheff envelope is determined by the intersection of the increasing branch $w x_j$ of the right point and the decreasing branch $(1-w)y_i$ of the left point. This occurs at
\[
  w_{ij}=\frac{y_i}{x_j+y_i},
  \qquad
  w_{ij}x_j=(1-w_{ij})y_i=\frac{x_jy_i}{x_j+y_i}.
\]
Integrating the piecewise-linear lower envelope (cf. Figure~\ref{fig:tchebycheff-shadows}) and collecting the terms gives the following formula; Appendix~\ref{app:r2-integral-derivation} spells out this elementary calculation in detail.
\begin{equation}
\boxed{
  \Rtwo(i_1,\ldots,i_k)
  =\frac{x_{i_1}}{2}
   +\sum_{r=1}^{k-1} A_{i_r i_{r+1}}
   +\frac{y_{i_k}}{2}
   -\sum_{r=1}^k A_{i_r i_r}.
}
\label{eq:adjacent-decomposition}
\end{equation}
Thus the exact integral $R_2$ subset objective is an ordered-subset functional: there are left and right boundary terms, unary diagonal corrections, and adjacent selected-neighbor transitions. This adjacent-neighbor decomposition is the key reason why Bellman dynamic programming works here.

\section{Exact algorithms}
\label{sec:algorithms}
This section presents the three exact dynamic-programming algorithms used in the paper. All three solve the same recurrence derived from the adjacent-neighbor decomposition. They differ only in how the predecessor minimization in one dynamic-programming layer is performed. Throughout the section, let
\[
  A_{ij}=\frac{x_jy_i}{2(x_j+y_i)},\qquad B_j=A_{jj},
\]
and let $P[r,j]$ denote the predecessor stored for backtracking.

\subsection{Direct Bellman dynamic programming}
\label{sec:bellman-dp}
\label{sec:direct-dp}
The adjacent-neighbor decomposition gives an immediate exact algorithm: build the selected sequence from left to right, and remember only the last selected point. This gives the following Bellman recurrence.

Let $D_r(j)$ denote the best partial value of a selected subsequence of exactly $r$ points that ends in point $j$, not yet including the final right boundary term $y_j/2$. From \eqref{eq:adjacent-decomposition}, the initialization is
\begin{equation}
  D_1(j)=\frac{x_j}{2}-A_{jj},
  \qquad j=1,\ldots,n.
\end{equation}
For $r=2,\ldots,k$ and $j=r,\ldots,n$, the Bellman recurrence is
\begin{equation}
  D_r(j)
  = -A_{jj}+\min_{i<j}\bigl\{D_{r-1}(i)+A_{ij}\bigr\}.
  \label{eq:dp}
\end{equation}
Finally,
\begin{equation}
  \Rtwo^*(k)=\min_{j\ge k}\left\{D_k(j)+\frac{y_j}{2}\right\}.
  \label{eq:dp-final}
\end{equation}
Storing the minimizing predecessor in \eqref{eq:dp} gives the selected subset by ordinary backtracking.

\begin{algorithm}[H]
\caption{Direct Bellman dynamic program for exact integral $\Rtwo$ subset selection}
\label{alg:direct-lrdp}
\small
\begin{tabbing}
\qquad\=\qquad\=\qquad\=\qquad\=\qquad\=\qquad\=\kill
\textbf{Direct-LRDP}$(p_1,\ldots,p_n,k)$\\
\>\textbf{for} $j=1,\ldots,n$ \textbf{do}\\
\>\> $D_1(j)\leftarrow x_j/2-B_j$\\
\>\textbf{for} $r=2,\ldots,k$ \textbf{do}\\
\>\>\textbf{for} $j=r,\ldots,n$ \textbf{do}\\
\>\>\> $D_r(j)\leftarrow +\infty$\\
\>\>\>\textbf{for} $i=r-1,\ldots,j-1$ \textbf{do}\\
\>\>\>\> $v\leftarrow D_{r-1}(i)+A_{ij}-B_j$\\
\>\>\>\>\textbf{if} $v<D_r(j)$ \textbf{then} $D_r(j)\leftarrow v$, $P[r,j]\leftarrow i$\\
\> $j^\star\leftarrow \arg\min_{j\ge k}\{D_k(j)+y_j/2\}$, with leftmost tie-breaking\\
\>\textbf{return} value $D_k(j^\star)+y_{j^\star}/2$ and the backtracked indices
\end{tabbing}
\end{algorithm}

\paragraph{Walkthrough.}
The first layer chooses one endpoint $j$ and pays the left boundary term and the diagonal correction. Layer $r$ constructs all best subsequences of cardinality $r$ ending at $j$ by trying every feasible predecessor $i<j$. The final step adds the right boundary term $y_j/2$ and selects the best last endpoint. This is the most transparent implementation of the recurrence and is useful as a correctness reference, but it scans $O(n^2)$ predecessor--endpoint pairs per layer.

\begin{theorem}[Correctness of the direct Bellman dynamic program]
Algorithm~\ref{alg:direct-lrdp} returns an optimum cardinality-$k$ subset for the exact integral $\Rtwo$ objective.
\end{theorem}
\begin{proof}
The proof is by induction over the selected cardinality $r$. For $r=1$, the value $D_1(j)=x_j/2-A_{jj}$ is exactly the left boundary term and diagonal correction for the one-point partial sequence ending at $j$, without the final right boundary term. Assume that $D_{r-1}(i)$ is the optimum value of every feasible partial sequence of cardinality $r-1$ ending at $i$. By the adjacent-neighbor decomposition, extending such a sequence by endpoint $j>i$ adds precisely the transition $A_{ij}$ and the new diagonal correction $-A_{jj}$. Taking the minimum over all feasible predecessors therefore gives the best cardinality-$r$ partial sequence ending at $j$. The final minimization adds the missing right boundary term $y_j/2$, so backtracking the stored predecessors yields an optimum selected sequence.
\end{proof}

\paragraph{Complexity.}
The direct implementation has $k$ layers and $O(n^2)$ transitions per layer, hence time complexity $O(kn^2)$ and memory $O(kn)$ with backpointers. A memory-reduced value-only version uses $O(n)$ memory.

\subsection{Monge transition structure}
\label{sec:monge}

Further DP speed-ups use Monge structure, which we discuss next.
Note that no objective normalization to a common range is required for the Monge argument below. What is needed is positivity after translation relative to the chosen utopian point and the sorted Pareto-front representation. Positive coordinate rescalings, for example $x\mapsto \alpha x$ and $y\mapsto \beta y$ with $\alpha,\beta>0$, change the numerical value of the indicator but preserve the Monge property of the transition matrix.

The transition cost \eqref{eq:transition} is generated by
\[
  \phi(x,y)=\frac{xy}{2(x+y)}.
\]
A direct calculation gives
\[
  \frac{\partial^2 \phi}{\partial x\,\partial y}
  =\frac{xy}{(x+y)^3}>0.
\]
Thus $\phi$ has increasing differences on the positive quadrant. The sign of this mixed derivative is the analytic property needed for the Monge proof. Consequently, positive multiplicative coordinate rescalings preserve the argument: replacing $\phi(x,y)$ by $\phi(\alpha x,\beta y)$ with $\alpha,\beta>0$ keeps the mixed derivative positive. Note that scaling may be useful to express preferences, but it is not required for the Monge condition.

Since the front order has $x_j<x_{j'}$ but $y_i>y_{i'}$ whenever $i<i'$ and $j<j'$, the transition matrix satisfies
\begin{equation}
  A_{ij}+A_{i'j'}\le A_{ij'}+A_{i'j},
  \qquad i<i',\; j<j'.
  \label{eq:monge}
\end{equation}
This is the Monge inequality in the index order used by the archive.

Consequently, the adjacent transition term
\[
  (i_r,i_{r+1})\mapsto A_{i_r i_{r+1}}
\]
is submodular on the two-coordinate componentwise lattice of increasing index vectors. Summing over adjacent ranks and adding unary and boundary terms preserves submodularity. Therefore, the fixed-cardinality exact integral $R_2$ subset problem also admits a submodular-lattice interpretation. This mirrors the Monge-to-lattice-to-mincut strategy in the one-dimensional Riesz-energy subset-selection work of Emmerich \cite{emmerich2026riesz}; there the full Riesz energy contains all selected-pair interactions, so a simple left-to-right Bellman principle is not exact. Here, the exact $R_2$ decomposition is local along the selected subsequence, so the Bellman recurrence is the most direct polynomial-time algorithm.

\label{sec:faster-algorithms}

The recurrence \eqref{eq:dp} can be accelerated because the transition costs form a Monge matrix. We first record the standard monotone-predecessor consequence used by the divide-and-conquer implementation. For a fixed layer $r$, define the implicit transition matrix
\begin{equation}
  M^{(r)}_{ij}=D_{r-1}(i)+A_{ij},
  \qquad i<j.
  \label{eq:layer-matrix}
\end{equation}
Adding the predecessor value $D_{r-1}(i)$ to row $i$ does not change any Monge difference, so $M^{(r)}$ satisfies the same inequality as $A$. Therefore, if ties are resolved by choosing the smallest predecessor index, the minimizer
\[
  \operatorname{opt}_r(j)
  =\argmin_{i<j}\{D_{r-1}(i)+A_{ij}\}
\]
is nondecreasing in $j$. Indeed, if $j<j'$ and two candidate predecessors satisfy $i<i'$, the Monge inequality gives
\[
  M^{(r)}_{ij}+M^{(r)}_{i'j'}
  \le
  M^{(r)}_{ij'}+M^{(r)}_{i'j}.
\]
This is the standard exchange inequality behind monotone matrix searching: a later column cannot force the leftmost optimum to move backwards.

\subsection{Divide-and-conquer dynamic programming}
\label{sec:divide-conquer}

The monotonicity of $\operatorname{opt}_r(j)$ gives the usual divide-and-conquer computation of one DP layer. To compute $D_r(j)$ for $j$ in an interval $[L,R]$, take the midpoint $m$, scan only the current admissible predecessor range $[a,b]$ to find $\operatorname{opt}_r(m)$, and recurse on
\[
  [L,m-1]\times [a,\operatorname{opt}_r(m)]
  \quad\text{and}\quad
  [m+1,R]\times [\operatorname{opt}_r(m),b].
\]
The triangular condition $i<j$ is handled by replacing the upper scan limit with $\min\{b,m-1\}$. This gives an $O(n\log n)$ worst-case bound for one layer with this simple recursive search. Since there are $k$ selected-cardinality layers, the full divide-and-conquer dynamic program has time bound
\[
  O(kn\log n).
\]
The reference implementation in the reproducibility repository includes both the direct $O(kn^2)$ method and this divide-and-conquer version, and checks that they return the same selected subset and value.

\begin{algorithm}[H]
\caption{Divide-and-conquer dynamic program using monotone predecessors}
\label{alg:divide-conquer-dp}
\small
\begin{tabbing}
\qquad\=\qquad\=\qquad\=\qquad\=\qquad\=\qquad\=\kill
\textbf{DCDP-Layer}$(r,L,R,a,b)$\\
\>\textbf{if} $L>R$ \textbf{return}\\
\> $m\leftarrow \lfloor(L+R)/2\rfloor$\\
\> $u\leftarrow \min\{b,m-1\}$\\
\> find the leftmost $i^\star\in\{a,\ldots,u\}$ minimizing $D_{r-1}(i)+A_{im}-B_m$\\
\> $D_r(m)\leftarrow D_{r-1}(i^\star)+A_{i^\star m}-B_m$, $P[r,m]\leftarrow i^\star$\\
\> \textbf{DCDP-Layer}$(r,L,m-1,a,i^\star)$\\
\> \textbf{DCDP-Layer}$(r,m+1,R,i^\star,b)$\\[1mm]
\textbf{Divide-Conquer-DP}$(p_1,\ldots,p_n,k)$\\
\> initialize $D_1(j)\leftarrow x_j/2-B_j$ for $j=1,\ldots,n$\\
\>\textbf{for} $r=2,\ldots,k$ \textbf{do}\\
\>\> call \textbf{DCDP-Layer}$(r,r,n,r-1,n-1)$\\
\> finish by minimizing $D_k(j)+y_j/2$ and backtracking
\end{tabbing}
\end{algorithm}

\paragraph{Walkthrough.}
The Monge property implies that the leftmost optimal predecessor is nondecreasing as the endpoint $j$ moves from left to right. The divide-and-conquer algorithm evaluates the middle endpoint of an interval first. Once the optimal predecessor for the middle endpoint is known, all endpoints to the left need only search predecessors no larger than that value, and all endpoints to the right need only search predecessors no smaller than that value. This recursively shrinks the search ranges. The recurrence values are exactly the same as in the direct Bellman dynamic program; only the order and range of predecessor scans change.

\begin{theorem}[Correctness and complexity of divide-and-conquer DP]
The divide-and-conquer dynamic program computes the same values as the direct Bellman recurrence. Its running time is $O(kn\log n)$ and its memory usage is $O(kn)$ with backpointers, or $O(n)$ for values only.
\end{theorem}
\begin{proof}
The Monge inequality and leftmost tie-breaking imply that the optimal predecessor index is nondecreasing in the endpoint $j$. Therefore, after the midpoint of a column interval has been solved, the predecessor search range for the left recursive subproblem can be restricted to predecessors no larger than the midpoint optimum, and the search range for the right recursive subproblem can be restricted to predecessors no smaller than it. These restrictions do not remove any optimal predecessor, so the values are identical to those of the direct recurrence. In one layer, each recursion level scans $O(n)$ candidate predecessor--endpoint pairs in total, and there are $O(\log n)$ levels. Thus one layer costs $O(n\log n)$, and $k$ layers cost $O(kn\log n)$.
\end{proof}

\subsection{Matrix-search dynamic programming}
\label{sec:matrix-search}
The divide-and-conquer method uses only the monotonicity of the minimizing predecessor. The same Monge structure gives a stronger result: each dynamic-programming layer can be computed in linear time. This can be viewed as a staircase version of monotone matrix search. Classical matrix-search algorithms, including SMAWK, compute minima of totally monotone matrices in linear time under implicit matrix access \cite{aggarwal1987smawk}. Here the feasible matrix is not rectangular, because column $j$ only permits predecessors $i<j$. We therefore give a direct lower-envelope version of the matrix search, specialized to the analytic transition functions of integral $R_2$.

\subsubsection{Preconditions}
For a fixed layer $r$, recall the implicit transition matrix
\[
  M^{(r)}_{ij}=D_{r-1}(i)+A_{ij}.
\]
The feasible domain is the staircase
\[
  \Omega_r=\{(i,j): r-1\le i<j\le n\}.
\]
Thus the feasible rows of column $j$ form the prefix $r-1,\ldots,j-1$. The matrix-search argument uses four preconditions: entries can be compared in constant time; feasible rows are nested prefixes; every feasible $2\times2$ submatrix satisfies the Monge inequality; and ties are resolved by the smallest predecessor index. Padding infeasible entries by $+\infty$ is not the right formal model, because it may introduce artificial ties. The correct object is a staircase partial matrix.

\subsubsection{Algorithm description}
For a fixed predecessor $i$, define
\begin{equation}
  f_i(x)=D_{r-1}(i)+\frac{x y_i}{2(x+y_i)}.
  \label{eq:matrix-search-curve}
\end{equation}
Then the transition part of the Bellman recurrence is
\begin{equation}
  \min_{i<j}\{D_{r-1}(i)+A_{ij}\}
  =
  \min_{i<j} f_i(x_j).
  \label{eq:matrix-search-query}
\end{equation}
The layer algorithm sweeps $j=r,r+1,\ldots,n$ from left to right. Before answering the query at $x_j$, it inserts the newly feasible predecessor $i=j-1$. The active predecessors are stored as the lower envelope of the functions $f_i$. The envelope is represented by a stack of rows together with crossing thresholds. A pointer moves forward along this stack as the query abscissae $x_j$ increase.

For two rows $i<i'$, set $\Delta=D_{r-1}(i')-D_{r-1}(i)$. Since $y_i>y_{i'}$, the crossing at which the later row $i'$ becomes as good as row $i$ is determined by
\begin{equation}
  \frac{x^2(y_i-y_{i'})}{2(x+y_i)(x+y_{i'})}=\Delta .
  \label{eq:crossing-equation}
\end{equation}
If $\Delta\le0$, the later row is already no worse at the origin and the crossing threshold is $0$. If $\Delta\ge (y_i-y_{i'})/2$, the later row never becomes strictly better for finite $x$. Otherwise, \eqref{eq:crossing-equation} has a unique positive solution, obtained from a quadratic equation. At the crossing itself the older row is kept, matching leftmost tie-breaking.

When a new row is inserted, its crossing with the last row on the stack is computed. If the new row never becomes better, it is discarded. If it becomes better before or at the point where the last row would start its own interval of strict optimality, the last row has no remaining open interval on the envelope and is popped. The test is repeated until the new row can be pushed or discarded. At a query $x_j$, the pointer is advanced while the next envelope row is strictly better at $x_j$. The current row is then the minimizing predecessor for column $j$.

\begin{algorithm}[H]
\caption{Matrix-search dynamic program by lower-envelope sweep}
\label{alg:matrix-search-dp}
\small
\begin{tabbing}
\qquad\=\qquad\=\qquad\=\qquad\=\qquad\=\qquad\=\kill
\textbf{Cross}$(i,h)$\\
\> return the first abscissa $x$ at which row $h$ is strictly better than row $i$\\
\> for $f_i(x)=D_{r-1}(i)+xy_i/(2(x+y_i))$; return $+\infty$ if this never occurs\\[1mm]
\textbf{Envelope-Layer}$(r)$\\
\> initialize an empty stack of pairs $(i,\alpha)$, where row $i$ is active from $x=\alpha$\\
\> initialize the query pointer to the first stack element\\
\>\textbf{for} $j=r,\ldots,n$ \textbf{do}\\
\>\> $h\leftarrow j-1$ \hfill\textit{new predecessor row entering the staircase}\\
\>\> $\alpha\leftarrow -\infty$\\
\>\>\textbf{while} stack is nonempty \textbf{do}\\
\>\>\> let $(i,\beta)$ be the last stack element\\
\>\>\> $\alpha\leftarrow \textbf{Cross}(i,h)$\\
\>\>\>\textbf{if} $\alpha\le\beta$ \textbf{then} pop the last stack element \textbf{else break}\\
\>\>\textbf{if} stack is empty \textbf{then} push $(h,-\infty)$\\
\>\>\textbf{else if} $\alpha<+\infty$ \textbf{then} push $(h,\alpha)$\\
\>\> advance the query pointer while the next stack interval starts at or before $x_j$\\
\>\> let $i^\star$ be the row at the query pointer\\
\>\> $D_r(j)\leftarrow D_{r-1}(i^\star)+A_{i^\star j}-B_j$, $P[r,j]\leftarrow i^\star$\\[1mm]
\textbf{Matrix-Search-DP}$(p_1,\ldots,p_n,k)$\\
\> initialize $D_1(j)\leftarrow x_j/2-B_j$ for $j=1,\ldots,n$\\
\>\textbf{for} $r=2,\ldots,k$ \textbf{do} call \textbf{Envelope-Layer}$(r)$\\
\> finish by minimizing $D_k(j)+y_j/2$ and backtracking
\end{tabbing}
\end{algorithm}

\paragraph{Walkthrough.}
In one layer, every feasible predecessor $i$ defines a curve
\[
  f_i(x)=D_{r-1}(i)+\frac{xy_i}{2(x+y_i)}.
\]
Endpoint $j$ queries the lower envelope of the active curves at $x_j$. As $j$ increases, one new predecessor row $h=j-1$ enters the feasible set. Because any two curves cross at most once, the lower envelope can be stored as an ordered stack of rows together with the abscissa from which each row becomes active. When a new row enters, it may destroy a suffix of the old envelope; those rows are popped. It is then inserted at its first strict crossing point with the surviving last row. Since both row insertion and query positions move from left to right, each row is pushed once and popped at most once, and the query pointer never moves backwards. Thus one layer takes linear time while computing the same Bellman values as the direct recurrence.

\subsubsection{Correctness}
The proof combines total monotonicity with the single-crossing structure of the functions \eqref{eq:matrix-search-curve}.

\begin{lemma}[Staircase total monotonicity]
For a fixed layer $r$, consider rows $i<i'$ and columns $j<j'$ such that all four entries belong to $\Omega_r$. If
\[
  M^{(r)}_{i'j}\le M^{(r)}_{ij},
\]
then
\[
  M^{(r)}_{i'j'}\le M^{(r)}_{ij'}.
\]
\end{lemma}
\begin{proof}
The row offsets $D_{r-1}(i)$ preserve the Monge inequality. Hence
\[
  M^{(r)}_{ij}+M^{(r)}_{i'j'}
  \le
  M^{(r)}_{ij'}+M^{(r)}_{i'j}.
\]
Rearranging gives
\[
  M^{(r)}_{i'j'}-M^{(r)}_{ij'}
  \le
  M^{(r)}_{i'j}-M^{(r)}_{ij}.
\]
The right-hand side is nonpositive by assumption. Thus the later row remains no worse in the later column. Only feasible entries of the staircase matrix are used.
\end{proof}

\begin{lemma}[Single crossing]
For $i<i'$, the difference $f_i(x)-f_{i'}(x)$ is strictly increasing for $x>0$.
\end{lemma}
\begin{proof}
The derivative of the nonlinear part of $f_i$ is
\[
  \frac{d}{dx}\left(\frac{x y}{2(x+y)}\right)
  =
  \frac{y^2}{2(x+y)^2}.
\]
Since $i<i'$ implies $y_i>y_{i'}$, and $y/(x+y)$ is strictly increasing in $y$ for every fixed $x>0$, we have
\[
  \frac{y_i^2}{2(x+y_i)^2}
  >
  \frac{y_{i'}^2}{2(x+y_{i'})^2}.
\]
Therefore $(f_i-f_{i'} )'(x)>0$. Hence two rows cross at most once, and after a later row becomes better it remains better for all larger $x$.
\end{proof}

\begin{theorem}[Linear-time layer computation]
For any fixed layer $r$, all values $D_r(j)$, $j=r,\ldots,n$, can be computed from $D_{r-1}$ in $O(n)$ time, assuming constant-time arithmetic comparisons and crossing computations.
\end{theorem}
\begin{proof}
By the single-crossing lemma, each new row has at most one crossing with each row currently on the envelope. The stack insertion rule is the standard lower-envelope update for single-crossing functions: a row is popped exactly when the new row becomes better no later than the popped row would have become strictly optimal. After the insertion terminates, the stack stores the active envelope rows in increasing predecessor order and the crossing thresholds are increasing. Because the query points $x_j$ are increasing, the query pointer never moves left. Each row is inserted once and popped at most once, and the pointer advances at most once per surviving envelope interval. Thus the total work in one layer is $O(n)$. The row returned at query $x_j$ minimizes \eqref{eq:matrix-search-query}; subtracting $A_{jj}$ gives the Bellman value \eqref{eq:dp}. Since switching occurs only when a later row is strictly better, ties are resolved by the smallest predecessor index.
\end{proof}

\subsubsection{Complexity}
Applying the theorem to all $k-1$ nontrivial layers gives an exact $O(kn)$ time algorithm for fixed-cardinality integral $R_2$ subset selection in the bi-objective setting. With backpointers the memory usage is $O(kn)$; if only the optimum value is needed, two value layers suffice and memory is $O(n)$. The reference implementation in the reproducibility repository includes this method. It uses floating-point arithmetic with small tolerances and checks it against the direct and divide-and-conquer dynamic programs on the reproducibility cases.

\paragraph{Relation to two-dimensional hypervolume subset selection.}
A similar matrix-search viewpoint also applies to the classical two-dimensional hypervolume subset-selection problem. For a sorted biobjective Pareto-front approximation, the hypervolume contribution of a selected sequence admits a local dynamic-programming recurrence with transition term of the bilinear form $-x_i y_j$, up to row and column offsets. This transition matrix is Monge or anti-Monge, depending on whether the problem is written as a minimization or maximization recurrence, and therefore the same monotone-predecessor/matrix-search principle yields $O(kn)$ time after sorting. This does not improve the best known two-dimensional hypervolume subset-selection bound $O((n-k)k+n\log n)$ \cite{bringmann2014twodimensional,kuhn2016hypervolume}, but it shows that the present integral $\Rtwo$ result fits a broader pattern: in two objectives, exact indicator subset selection often becomes fast when the indicator decomposes into local predecessor-successor terms with a Monge transition structure.

\subsection{Complexity summary for the three algorithms}
\label{subsec:complexity-scaling}

The improvement should be interpreted per dynamic-programming layer. The direct recurrence spends $O(n^2)$ time in each layer and therefore $O(kn^2)$ time overall. Monge monotonicity reduces one layer to $O(n\log n)$ by divide-and-conquer search, giving $O(kn\log n)$ total time. Thus, for fixed or small $k$, the accelerated method is nearly linear in $n$ up to a logarithmic factor. For balanced subset sizes, however, $k=\Theta(n)$, so the same bound becomes
\[
  O(n^2\log n).
\]
It is therefore not correct to describe the full algorithm simply as $O(n\log n)$ unless $k$ is treated as a fixed constant.

The asymptotic regimes are summarized in Table~\ref{tab:scaling-regimes}. 

\begin{table}[H]
\centering
\begin{tabular}{lccc}
\toprule
Regime & Direct DP & Divide-and-conquer DP & Matrix-search DP \\
\midrule
fixed $k$ & $O(n^2)$ & $O(n\log n)$ & $O(n)$ \\
$k=\Theta(n)$ & $O(n^3)$ & $O(n^2\log n)$ & $O(n^2)$ \\
\bottomrule
\end{tabular}
\caption{Complexity interpretation in two common asymptotic regimes. The divide-and-conquer bound is $O(n\log n)$ per layer and $O(kn\log n)$ overall; the matrix-search bound is $O(n)$ per layer and $O(kn)$ overall.}
\label{tab:scaling-regimes}
\end{table}

\section{Computational study}
\label{sec:computational-study}
The experiments have two purposes. First, they verify that the three dynamic-programming implementations agree with exhaustive enumeration on instances where exhaustive enumeration is feasible. Second, they illustrate the empirical scaling of the direct, divide-and-conquer, and matrix-search implementations in balanced and fixed-cardinality regimes.
\subsection{Illustrative seven-point instance}

Figure~\ref{fig:front} shows the seven-point staircase instance used in the Python test file. For $k=3$, the dynamic program selects points $1,5,7$.

\begin{figure}[H]
\centering
\begin{tikzpicture}[scale=0.27, every node/.style={font=\small}]
  \draw[-{Latex[length=2.4mm]}] (0,0) -- (19,0) node[right] {$f_1$};
  \draw[-{Latex[length=2.4mm]}] (0,0) -- (0,22) node[above] {$f_2$};
  \foreach \x in {0,5,10,15} {\draw[gray!35] (\x,0) -- (\x,21.2); \node[below] at (\x,0) {\x};}
  \foreach \y in {0,5,10,15,20} {\draw[gray!35] (0,\y) -- (18.2,\y); \node[left] at (0,\y) {\y};}

  \draw[gray!70, thick] (2,20) -- (4,18) -- (6,16) -- (9,12) -- (11,8) -- (14,5) -- (17,3);

  \foreach \x/\y/\lab in {2/20/1,4/18/2,6/16/3,9/12/4,11/8/5,14/5/6,17/3/7} {
    \fill[gray!70] (\x,\y) circle (3pt);
    \node[above right=1pt] at (\x,\y) {$p_{\lab}$};
  }

  \foreach \x/\y/\lab in {2/20/1,11/8/5,17/3/7} {
    \draw[red!70!black, very thick] (\x,\y) circle (6pt);
  }
  \draw[red!70!black, very thick, dashed] (2,20) -- (11,8) -- (17,3);
  \node[align=left, anchor=west] at (1,2) {Selected for $k=3$:\\$\{p_1,p_5,p_7\}$};
\end{tikzpicture}
\caption{A sorted bi-objective Pareto front. The red-circled points are the optimum cardinality-three subset under the exact integral $R_2$ objective for this instance.}
\label{fig:front}
\end{figure}
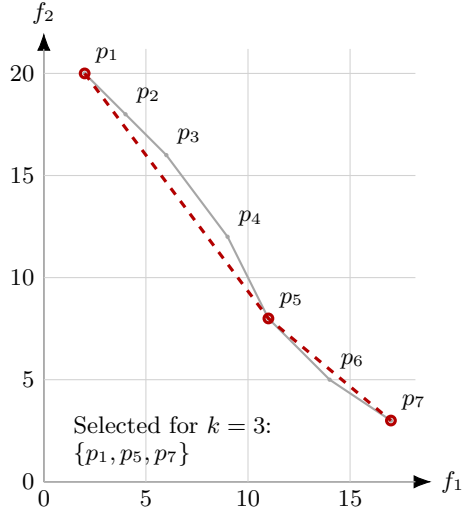

\subsection{Illustrative twenty-point instance}

Figure~\ref{fig:larger-front} shows a twenty-point Pareto-front approximation. For $k=10$, the dynamic program selects points
\[
  \{p_1,p_4,p_6,p_7,p_8,p_9,p_{10},p_{12},p_{14},p_{17}\},
\]
with exact value $\Rtwo=4.936582464$. This instance is also included in the reproducibility material and is verified by exhaustive enumeration.

\begin{figure}[H]
\centering
\begin{tikzpicture}[x=0.118cm,y=0.155cm, every node/.style={font=\scriptsize}]
  \draw[-{Latex[length=2.4mm]}] (0,0) -- (99,0) node[right] {$f_1$};
  \draw[-{Latex[length=2.4mm]}] (0,0) -- (0,43) node[above] {$f_2$};
  \foreach \x in {0,20,40,60,80} {\draw[gray!30] (\x,0) -- (\x,42); \node[below] at (\x,0) {\x};}
  \foreach \y in {0,10,20,30,40} {\draw[gray!30] (0,\y) -- (96,\y); \node[left] at (0,\y) {\y};}

  \draw[gray!65, thick]
    (1,40) -- (2,32) -- (3,28) -- (4,24) -- (5.5,20.5)
    -- (7,17.5) -- (9,15) -- (11.5,12.8) -- (14,10.8) -- (17,9.1)
    -- (20,7.6) -- (24,6.2) -- (29,5.1) -- (35,4.2) -- (42,3.4)
    -- (50,2.7) -- (59,2.15) -- (69,1.7) -- (80,1.35) -- (92,1.05);

  \foreach \x/\y/\lab in {1/40/1,2/32/2,3/28/3,4/24/4,5.5/20.5/5,7/17.5/6,9/15/7,11.5/12.8/8,14/10.8/9,17/9.1/10,20/7.6/11,24/6.2/12,29/5.1/13,35/4.2/14,42/3.4/15,50/2.7/16,59/2.15/17,69/1.7/18,80/1.35/19,92/1.05/20} {
    \fill[gray!65] (\x,\y) circle (2.2pt);
  }

  \draw[red!70!black, very thick, dashed]
    (1,40) -- (4,24) -- (7,17.5) -- (9,15) -- (11.5,12.8)
    -- (14,10.8) -- (17,9.1) -- (24,6.2) -- (35,4.2) -- (59,2.15);

  \foreach \x/\y/\lab in {1/40/1,4/24/4,7/17.5/6,9/15/7,11.5/12.8/8,14/10.8/9,17/9.1/10,24/6.2/12,35/4.2/14,59/2.15/17} {
    \draw[red!70!black, very thick] (\x,\y) circle (4.2pt);
    \node[red!70!black, above right=1pt] at (\x,\y) {$p_{\lab}$};
  }

  \node[align=left, anchor=north west, fill=white, draw=gray!40, rounded corners=2pt, inner sep=3pt]
    at (51,41) {Twenty candidate points\\selected subset for $k=10$\\$\Rtwo=4.936582464$};
\end{tikzpicture}
\caption{A larger Pareto-front approximation with twenty candidate points. The red-circled points are the optimum cardinality-ten subset under the exact integral $R_2$ objective for this instance.}
\label{fig:larger-front}
\end{figure}
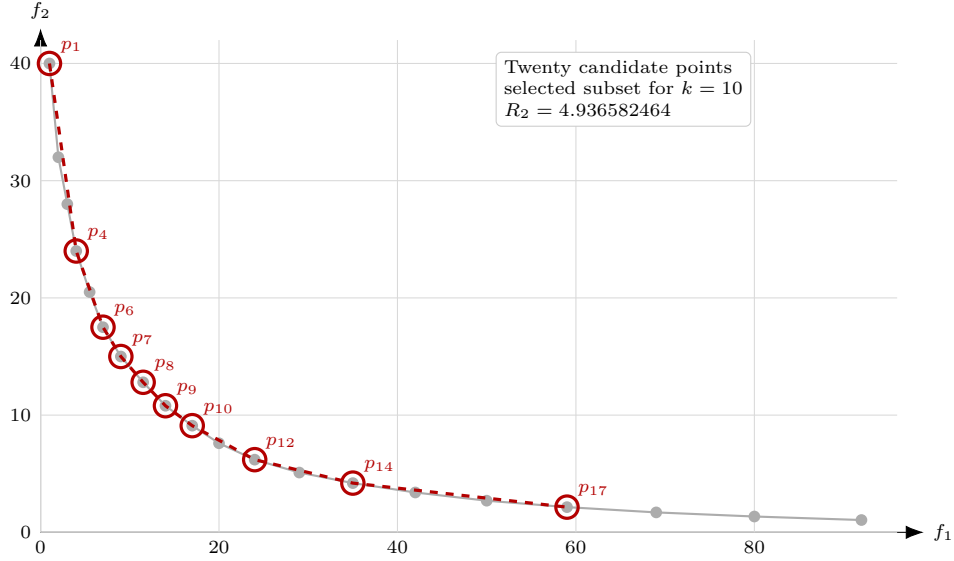

\subsection{Empirical verification tests}
\label{sec:examples}

The reproducibility material in the repository implements the direct recurrence \eqref{eq:dp}, the divide-and-conquer accelerated recurrence, the matrix-search recurrence from Section~\ref{sec:matrix-search}, exact value evaluation by \eqref{eq:adjacent-decomposition}, and brute-force verification for the displayed examples. Representative output is shown in Table~\ref{tab:tests}. The same material also includes the runtime experiment summarized in Section~\ref{sec:runtime}. The final two blocks add a twelve-point instance with larger cardinalities and a twenty-point graphical instance; in both cases, the dynamic-programming result is checked against complete enumeration. The repository also includes an eighty-point non-brute-force check comparing the direct, divide-and-conquer, and matrix-search implementations.

\begin{table}[H]
\centering
\small
\begin{tabular}{llp{6.4cm}r}
\toprule
Instance & $k$ & selected one-based indices & exact $R_2$ \\
\midrule
Seven-point staircase & 2 & $(1,6)$ & 4.866450887 \\
Seven-point staircase & 3 & $(1,5,7)$ & 4.268506714 \\
Seven-point staircase & 4 & $(1,3,5,7)$ & 4.105253002 \\
Seven-point staircase & 5 & $(1,3,5,6,7)$ & 4.020420466 \\
Six-point mildly irregular & 3 & $(1,4,6)$ & 2.824432278 \\
Five-point convex & 3 & $(2,4,5)$ & 2.783431481 \\
Twelve-point larger-$k$ front & 6 & $(2,4,5,6,8,11)$ & 3.674859356 \\
Twelve-point larger-$k$ front & 7 & $(2,4,5,6,7,9,12)$ & 3.603086767 \\
Twelve-point larger-$k$ front & 8 & $(1,3,4,5,6,7,9,12)$ & 3.546608278 \\
Twelve-point larger-$k$ front & 9 & $(1,3,4,5,6,7,8,10,12)$ & 3.501284795 \\
Twenty-point graphical front & 8 & $(1,4,6,7,9,11,13,16)$ & 5.035586573 \\
Twenty-point graphical front & 10 & $(1,4,6,7,8,9,10,12,14,17)$ & 4.936582464 \\
Twenty-point graphical front & 12 & $(1,3,5,6,7,8,9,10,11,12,14,17)$ & 4.863355081 \\
\bottomrule
\end{tabular}
\caption{All displayed dynamic-programming solutions are verified by exhaustive enumeration in the reproducibility material; the direct, divide-and-conquer, and matrix-search implementations also agree.}
\label{tab:tests}
\end{table}

\subsection{Runtime assessment with CPU times}
\label{sec:runtime}

To complement the asymptotic discussion, the reference implementation in the reproducibility repository measures CPU wall-clock times using \texttt{time.perf\_counter()} on deterministic fronts. The benchmark compares exhaustive enumeration, the direct left-to-right dynamic program (LRDP), the divide-and-conquer dynamic program (DC DP), and the matrix-search dynamic program. Exhaustive enumeration is run with a one-second wall-clock time limit per case. All dynamic-programming variants are checked against each other for every benchmark case. When exhaustive enumeration finishes within the time limit, its selected subset and value are also checked against the DP result; otherwise the table records the exhaustive entry as $>1000$ ms. The DP timings are averages over ten runs. The timings are intended as reproducibility-oriented reference figures from the unoptimized pure-Python code, not as tuned implementation benchmarks.

Table~\ref{tab:runtime-balanced} reports balanced instances with $k=n/2$. The exhaustive runtime grows quickly with the number of subsets, while all dynamic-programming variants remain below a millisecond throughout this range. Figure~\ref{fig:runtime-balanced} visualizes the same data on a logarithmic vertical axis. Time-limited exhaustive entries are shown at the one-second cap.

\begin{table}[H]
\centering
\scriptsize
\begin{tabular}{rrrrrrr}
\toprule
$n$ & $k$ & $\binom{n}{k}$ & Exhaustive (ms) & LRDP (ms) & DC DP (ms) & Matrix (ms) \\
\midrule
8 & 4 & 70 & 0.233 & 0.056 & 0.060 & 0.082 \\
10 & 5 & 252 & 0.983 & 0.061 & 0.058 & 0.099 \\
12 & 6 & 924 & 3.267 & 0.077 & 0.075 & 0.135 \\
14 & 7 & 3,432 & 13.505 & 0.107 & 0.104 & 0.181 \\
16 & 8 & 12,870 & 55.671 & 0.148 & 0.128 & 0.236 \\
18 & 9 & 48,620 & 229.944 & 0.205 & 0.170 & 0.309 \\
20 & 10 & 184,756 & $>1000$ & 0.364 & 0.274 & 0.457 \\
\bottomrule
\end{tabular}
\caption{Mean wall-clock runtimes in milliseconds for balanced cases $k=n/2$. Exhaustive enumeration has a one-second time limit per case.}
\label{tab:runtime-balanced}
\end{table}

\begin{figure}[H]
\centering
\begin{tikzpicture}
\begin{semilogyaxis}[
  width=0.70\textwidth,
  height=7.0cm,
  xlabel={number of candidate points $n$},
  ylabel={mean runtime (ms)},
  xmin=8, xmax=20,
  ymin=0.03, ymax=1500,
  xtick={8,10,12,14,16,18,20},
  grid=both,
  major grid style={gray!35},
  minor grid style={gray!18},
  legend style={at={(-0.22,0.98)}, anchor=north east, fill=white, draw=black, font=\scriptsize, cells={anchor=west}},
  tick align=outside
]
\addplot+[thick, mark=*, blue] coordinates {(8,0.233) (10,0.983) (12,3.267) (14,13.505) (16,55.671) (18,229.944)};
\addplot+[only marks, mark=x, mark size=3.0pt, blue] coordinates {(20,1000)};
\addplot+[thick, mark=square*, red!75!black] coordinates {(8,0.056) (10,0.061) (12,0.077) (14,0.107) (16,0.148) (18,0.205) (20,0.364)};
\addplot+[thick, mark=triangle*, green!50!black] coordinates {(8,0.060) (10,0.058) (12,0.075) (14,0.104) (16,0.128) (18,0.170) (20,0.274)};
\addplot+[thick, mark=diamond*, RTwoPurple] coordinates {(8,0.082) (10,0.099) (12,0.135) (14,0.181) (16,0.236) (18,0.309) (20,0.457)};
\legend{Exhaustive, Exhaustive time limit, LRDP, DC DP, Matrix DP}
\end{semilogyaxis}
\end{tikzpicture}
\caption{CPU runtime comparison on deterministic balanced fronts with $k=n/2$. The cross marks a case where exhaustive enumeration reached the one-second time limit.}
\label{fig:runtime-balanced}
\end{figure}
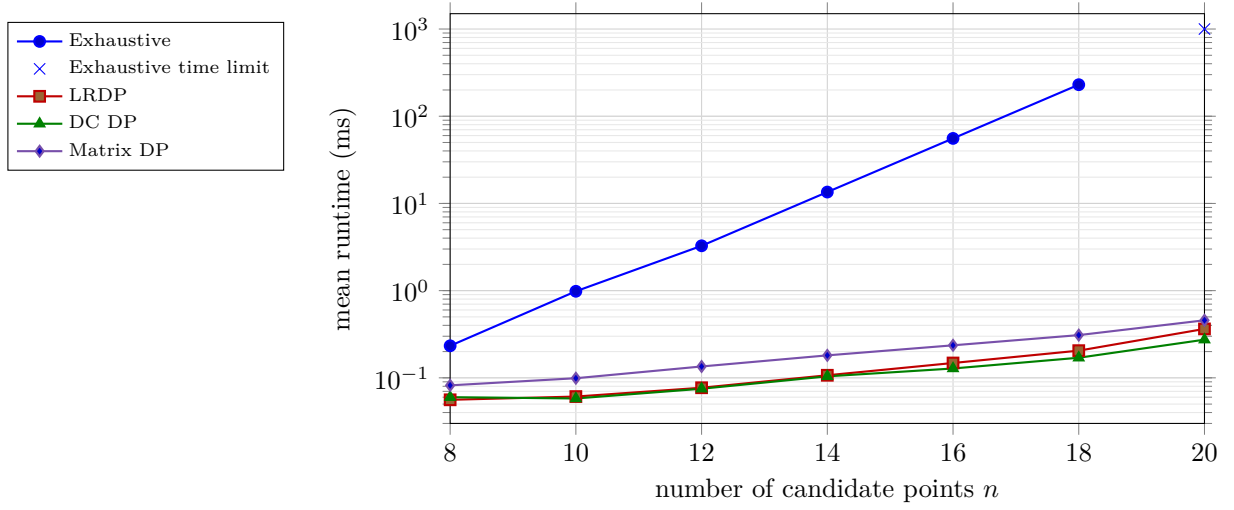

Table~\ref{tab:runtime-fixed-k} repeats the experiment for a small constant cardinality, here $k=6$, and extends the range to $n=100$. This illustrates the fixed-$k$ regime discussed in Table~\ref{tab:scaling-regimes}: exhaustive enumeration grows like $\binom{n}{6}=\Theta(n^6)$, while the dynamic programs grow gently. Figure~\ref{fig:runtime-fixed-k} gives the corresponding plot, with capped exhaustive entries shown at the one-second line.

\begin{table}[H]
\centering
\scriptsize
\begin{tabular}{rrrrrrr}
\toprule
$n$ & $k$ & $\binom{n}{k}$ & Exhaustive (ms) & LRDP (ms) & DC DP (ms) & Matrix (ms) \\
\midrule
10 & 6 & 210 & 0.775 & 0.058 & 0.066 & 0.116 \\
14 & 6 & 3,003 & 12.066 & 0.104 & 0.097 & 0.167 \\
18 & 6 & 18,564 & 67.555 & 0.156 & 0.229 & 0.411 \\
22 & 6 & 74,613 & 288.046 & 0.233 & 0.242 & 0.299 \\
26 & 6 & 230,230 & 870.697 & 0.306 & 0.190 & 0.339 \\
30 & 6 & 593,775 & $>1000$ & 0.411 & 0.235 & 0.390 \\
40 & 6 & 3,838,380 & $>1000$ & 0.751 & 0.438 & 0.553 \\
60 & 6 & 50,063,860 & $>1000$ & 1.625 & 0.513 & 1.005 \\
80 & 6 & 300,500,200 & $>1000$ & 2.858 & 1.128 & 1.053 \\
100 & 6 & 1,192,052,400 & $>1000$ & 4.408 & 0.922 & 1.327 \\
\bottomrule
\end{tabular}
\caption{Mean wall-clock runtimes in milliseconds for the fixed-cardinality regime $k=6$, extended to $n=100$. Exhaustive enumeration has a one-second time limit per case.}
\label{tab:runtime-fixed-k}
\end{table}

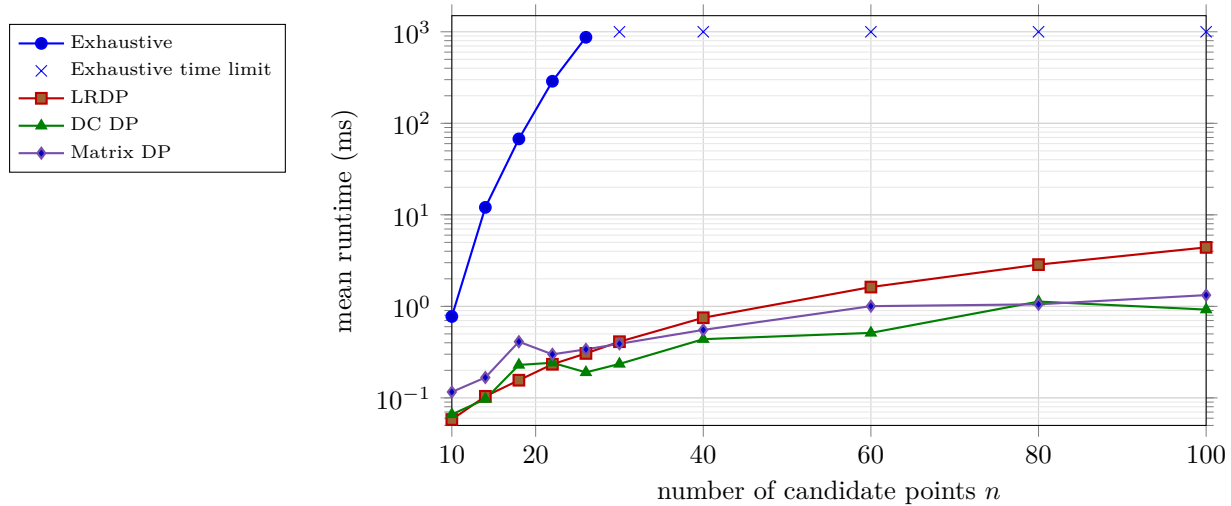
\begin{figure}[H]
\centering
\begin{tikzpicture}
\begin{semilogyaxis}[
  width=0.70\textwidth,
  height=7.0cm,
  xlabel={number of candidate points $n$},
  ylabel={mean runtime (ms)},
  xmin=10, xmax=100,
  ymin=0.05, ymax=1500,
  xtick={10,20,40,60,80,100},
  grid=both,
  major grid style={gray!35},
  minor grid style={gray!18},
  legend style={at={(-0.22,0.98)}, anchor=north east, fill=white, draw=black, font=\scriptsize, cells={anchor=west}},
  tick align=outside
]
\addplot+[thick, mark=*, blue] coordinates {(10,0.775) (14,12.066) (18,67.555) (22,288.046) (26,870.697)};
\addplot+[only marks, mark=x, mark size=3.0pt, blue] coordinates {(30,1000) (40,1000) (60,1000) (80,1000) (100,1000)};
\addplot+[thick, mark=square*, red!75!black] coordinates {(10,0.058) (14,0.104) (18,0.156) (22,0.233) (26,0.306) (30,0.411) (40,0.751) (60,1.625) (80,2.858) (100,4.408)};
\addplot+[thick, mark=triangle*, green!50!black] coordinates {(10,0.066) (14,0.097) (18,0.229) (22,0.242) (26,0.190) (30,0.235) (40,0.438) (60,0.513) (80,1.128) (100,0.922)};
\addplot+[thick, mark=diamond*, RTwoPurple] coordinates {(10,0.116) (14,0.167) (18,0.411) (22,0.299) (26,0.339) (30,0.390) (40,0.553) (60,1.005) (80,1.053) (100,1.327)};
\legend{Exhaustive, Exhaustive time limit, LRDP, DC DP, Matrix DP}
\end{semilogyaxis}
\end{tikzpicture}
\caption{CPU runtime comparison for a small constant cardinality, $k=6$, up to $n=100$. All three dynamic-programming implementations are checked against each other for every displayed case.}
\label{fig:runtime-fixed-k}
\end{figure}

\section{Summary and outlook}
\label{sec:conclusions}

This paper has derived exact dynamic-programming algorithms for fixed-cardinality subset selection under the bi-objective integral $R_2$ indicator. The central observation is that, for a sorted Pareto-front approximation, the continuous Tchebycheff-envelope integral decomposes into boundary terms, unary diagonal corrections, and adjacent selected-neighbor transitions. This gives a direct Bellman recurrence with $O(kn^2)$ running time. The transition matrix is Monge, so predecessor indices are monotone and a divide-and-conquer implementation reduces the cost to $O(kn\log n)$. The stronger matrix-search result uses the staircase totally monotone structure of each layer and a lower-envelope sweep to compute each layer in $O(n)$ time, giving an exact $O(kn)$ algorithm under constant-time arithmetic comparisons. The reference implementation in the reproducibility repository verifies the direct, divide-and-conquer, and matrix-search dynamic programs against exhaustive enumeration where feasible and reports benchmark data under explicit time limits.

The result should be interpreted in the exact integral setting. It does not rely on a finite weight-vector discretization, and no normalization to $[0,1]$ is required for the Monge condition. What is required is positivity relative to the chosen utopian point and a sorted representation of a two-dimensional Pareto-front approximation. Positive coordinate rescalings preserve the Monge proof, although they change the numerical indicator values and therefore remain a modeling choice.

Several directions remain open. First, other weight densities or scalarizing families may admit similar ordered-subset decompositions, but the transition formula and Monge proof would have to be derived separately. Second, the matrix-search proof suggests looking for other indicator subset-selection problems whose dynamic-programming layers are staircase Monge or totally monotone matrices. Third, the higher-dimensional exact integral $R_2$ setting studied by Jaszkiewicz and Zielniewicz suggests an important extension target. The present dynamic program is inherently bi-objective: the weight domain is the interval $[0,1]$, the lower envelope is one-dimensional and piecewise linear, and sorted Pareto-front approximations have a predecessor--successor structure. For three or more objectives, the weight domain is a simplex and the lower envelope becomes a higher-dimensional polyhedral complex. A selected point can interact through facets with many other points, so the adjacent-neighbor decomposition used here need not survive. Understanding whether special higher-dimensional Pareto-front geometries, fixed dimension, or computational-geometry representations of the Tchebycheff envelope lead to exact subset-selection algorithms is a natural open problem. Finally, the algorithm suggests a practical archive-reduction tool for exact integral $R_2$ benchmarking and decision-support pipelines; integrating it into multiobjective optimization software and comparing it with hypervolume-based subset selection are natural next experimental steps.

\section*{Code availability}
The reproducibility material is available at\par\noindent\url{https://github.com/emmerichmtm/integral-r2-subset-selection}.
The repository contains the Python reference implementation, the consistency tests comparing exhaustive enumeration and the dynamic programs, and the scripts used to generate the runtime tables.

\section*{Acknowledgements}
This research is related to the thematic research area Decision Analytics utilizing Causal Models and Multiobjective Optimization (DEMO, jyu.fi/demo) of the University of Jyvaskyla.

\section*{Declaration on the use of generative AI}
Generative AI tools were used for editorial assistance, LaTeX and Python drafting support, and consistency checking during manuscript preparation. All mathematical statements, proofs, algorithms, experiments, citations, and final wording were reviewed by the author, who remains fully responsible for the content of the manuscript.

\appendix

\newpage
\Needspace{10\baselineskip}
\section{Derivation of the exact bi-objective integral \texorpdfstring{$R_2$}{R2} decomposition}
\label{app:r2-integral-derivation}

This appendix gives a detailed derivation of the closed-form expression for the exact integral bi-objective
$R_2$ value used in the main text. It is included for exposition and for notational self-containment only. It does not add a new mathematical result. The integral viewpoint is due to the exact bi-objective $R_2$ analysis of Sch\"apermeier and Kerschke~\cite{schaepermeier2024reinvestigating,schaepermeier2025r2v2}; the Tchebycheff-shadow terminology and geometric interpretation used in Figure~\ref{fig:tchebycheff-shadows} follows Emmerich~\cite{tchebEmmerich}.

We use the notation of Figure~\ref{fig:tchebycheff-shadows}. Each archive point
\[
  p_i=(x_i,y_i)
\]
defines a Tchebycheff shadow over the weight interval,
\[
  \tau_i(w)=\max\{w x_i,(1-w)y_i\},
  \qquad 0\le w\le 1.
\]
For a selected subset $S$, the exact integral $R_2$ value is the area under the lower envelope of these shadows:
\[
  R_2(S)
  =
  \int_0^1 \min_{i\in S}\tau_i(w)\,dw
  =
  \int_0^1
  \min_{i\in S}
  \max\{w x_i,(1-w)y_i\}\,dw.
\]
Thus, for each weight $w$, one first chooses the selected point with smallest Tchebycheff loss, and only after this pointwise minimization is the result integrated over $w$.

Figure~\ref{fig:appendix-tchebycheff-shadow-geometry} further illustrates the terminology used below. In objective space, the weighted Tchebycheff sublevel sets are rectangles anchored at the translated utopian point. In weight space, the same point gives the shadow function $\tau_i(w)$; this is the object whose lower envelope is integrated.

\begin{figure}[H]
\centering
\resizebox{\textwidth}{!}{%
\begin{tikzpicture}[
  >=Latex,
  font=\small,
  shadow/.style={blue!65!black, very thick},
  otherShadow/.style={gray!65, thick, dashed},
  point/.style={circle, fill=gray!65, inner sep=1.8pt},
  focusPoint/.style={circle, fill=blue!65!black, inner sep=2.7pt},
  annotation/.style={align=left, font=\small},
]

% ================================================================
% Left panel: objective space
% ================================================================
\begin{scope}[shift={(0,0)}, x=0.62cm, y=0.62cm]
  \node[anchor=west] at (0.35,7.55) {\textbf{Objective space}};

  % Axes
  \draw[->] (0,0) -- (8.9,0) node[right] {$f_1$};
  \draw[->] (0,0) -- (0,7.2) node[above] {$f_2$};

  % Utopian point
  \fill[black] (0,0) circle (1.7pt);
  \node[below left] at (0,0) {$z^+$};

  % Example Pareto chain points
  \coordinate (pone) at (1.0,6.2);
  \coordinate (ptwo) at (1.8,5.3);
  \coordinate (pi) at (3.1,4.1);
  \coordinate (pfour) at (4.7,2.7);
  \coordinate (pfive) at (6.7,1.5);

  \draw[gray!55, thick] (pone) -- (ptwo) -- (pi) -- (pfour) -- (pfive);
  \foreach \p/\lab in {pone/{p_1},ptwo/{p_2},pfour/{p_4},pfive/{p_5}}{
    \node[point] at (\p) {};
    \node[above right=1pt] at (\p) {$\lab$};
  }
  \node[focusPoint] at (pi) {};
  \node[above right=2pt, blue!65!black] at (pi) {$p_i$};

  % Coordinate projections for p_i
  \draw[blue!65!black, densely dotted, thick] (pi) -- (3.1,0) node[below] {$x_i$};
  \draw[blue!65!black, densely dotted, thick] (pi) -- (0,4.1);
  \node[left, blue!65!black] at (0,4.1) {$y_i$};

  % Balanced Tchebycheff rectangle for p_i
  \draw[blue!65!black, thick, rounded corners=1pt]
    (0,0) rectangle (3.1,4.1);

  % Additional rectangles for intuition
  \draw[gray!45, dashed] (0,0) rectangle (2.0,5.8);
  \draw[gray!45, dashed] (0,0) rectangle (5.5,2.3);

  % Annotation box to the right, away from the chain
  \node[draw=blue!65!black, rounded corners=2pt, fill=white,
        annotation, text=blue!65!black, anchor=north west, inner sep=3pt]
    at (5.15,6.55)
    {$p_i=(x_i,y_i)$\\[-1mm]
     weighted Tchebycheff\\[-1mm]
     rectangle touching\\[-1mm]
     the highlighted point};

  \node[annotation, anchor=north west] at (0.2,-0.55)
    {For fixed $w$ and level $t$:\\[-1mm]
     $\max\{w x,(1-w)y\}\le t$ is a rectangle.};
\end{scope}

% Connector arrow
\draw[->, thick, blue!65!black] (6.15,2.55) -- (7.65,2.55)
  node[midway, below, align=center] {induces\\shadow};

% ================================================================
% Right panel: weight space
% ================================================================
\begin{scope}[shift={(8.15,0)}, x=6.1cm, y=0.85cm]
  \node[anchor=west] at (0,7.55) {\textbf{Weight space}};

  % Axes
  \draw[->] (0,0) -- (1.08,0) node[right] {$w$};
  \draw[->] (0,0) -- (0,6.9) node[above] {loss};
  \draw (0,-0.08) -- (0,0.08) node[below=3pt] {$0$};
  \draw (1,-0.08) -- (1,0.08) node[below=3pt] {$1$};

  % Other shadows: tau_j(w)=max{w x_j,(1-w)y_j}
  \draw[otherShadow] (0,6.2) -- ({6.2/(1.0+6.2)},{1.0*6.2/(1.0+6.2)}) -- (1,1.0);
  \draw[otherShadow] (0,5.3) -- ({5.3/(1.8+5.3)},{1.8*5.3/(1.8+5.3)}) -- (1,1.8);
  \draw[otherShadow] (0,2.7) -- ({2.7/(4.7+2.7)},{4.7*2.7/(4.7+2.7)}) -- (1,4.7);
  \draw[otherShadow] (0,1.5) -- ({1.5/(6.7+1.5)},{6.7*1.5/(6.7+1.5)}) -- (1,6.7);

  % Focus shadow p_i=(3.1,4.1)
  \pgfmathsetmacro{\xi}{3.1}
  \pgfmathsetmacro{\yi}{4.1}
  \pgfmathsetmacro{\ci}{\yi/(\xi+\yi)}
  \pgfmathsetmacro{\hi}{\xi*\yi/(\xi+\yi)}
  \draw[shadow] (0,\yi) -- (\ci,\hi) -- (1,\xi);
  \fill[blue!65!black] (\ci,\hi) circle (2.2pt);
  \draw[blue!65!black, densely dotted] (\ci,0) -- (\ci,\hi);
  \node[below, blue!65!black] at (\ci,0) {$c_i$};
  \node[right, blue!65!black] at (\ci,\hi) {kink};

  % Branch annotations
  \node[blue!65!black, anchor=east] at (0.38,3.76) {$(1-w)y_i$};
  \node[blue!65!black, anchor=west] at (0.75,2.95) {$w x_i$};

  % Equation box
  \node[draw=blue!65!black, rounded corners=2pt, fill=white,
        align=left, anchor=north west, inner sep=4pt]
    at (0.09,6.35)
    {$\displaystyle \tau_i(w)=\max\{w x_i,(1-w)y_i\}$\\[1mm]
     $\displaystyle c_i=\frac{y_i}{x_i+y_i}$};

  \node[annotation, anchor=north west] at (0.02,-0.55)
    {Dashed curves: shadows $\tau_j(w)$ of other points.\\
     Bold curve: highlighted shadow $\tau_i(w)$.};
\end{scope}

\end{tikzpicture}%
}
\caption{Objective-space and weight-space views of a Tchebycheff shadow. The left panel shows a sorted nondominated chain and weighted Tchebycheff sublevel rectangles anchored at the translated utopian point $z^+$. The right panel shows the corresponding shadow curves $\tau_j(w)$ over the weight interval. The highlighted point $p_i=(x_i,y_i)$ gives the bold shadow $\tau_i(w)=\max\{w x_i,(1-w)y_i\}$, with kink at $c_i=y_i/(x_i+y_i)$.}
\label{fig:appendix-tchebycheff-shadow-geometry}
\end{figure}
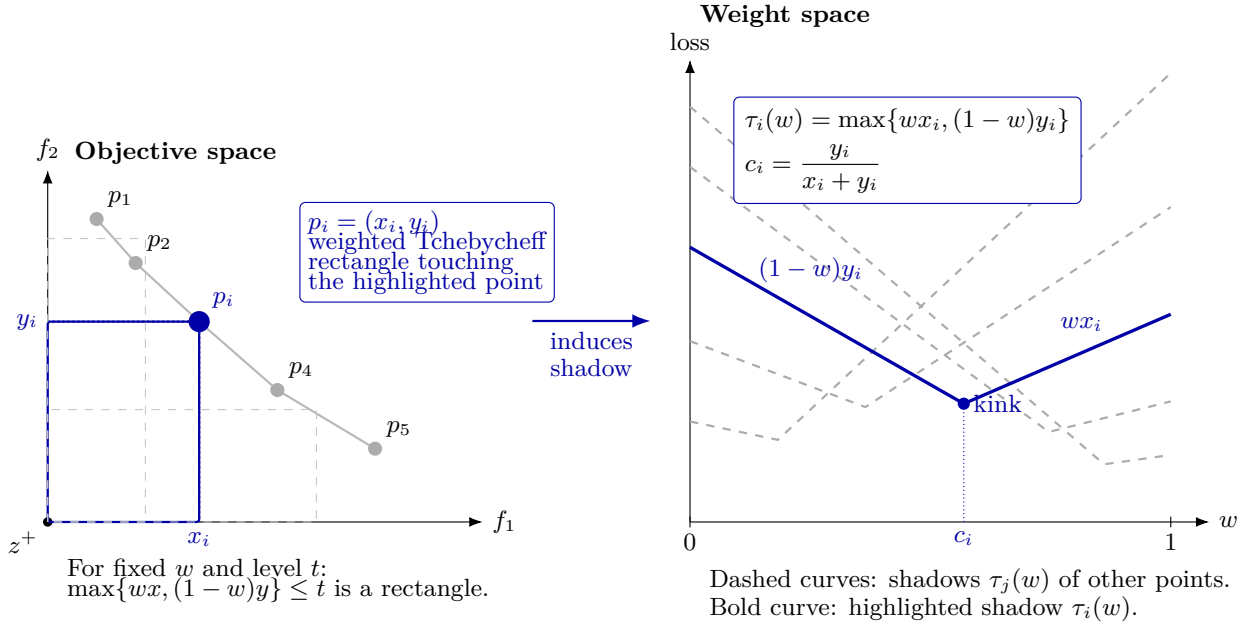

We assume throughout that the archive is a sorted nondominated chain,
\[
  0<x_1<\cdots <x_n,
  \qquad
  y_1>\cdots >y_n>0.
\]
Let the selected subset be written as
\[
  S=(i_1<i_2<\cdots<i_k).
\]
We shall derive
\[
  R_2(i_1,\ldots,i_k)
  =
  \frac{x_{i_1}}{2}
  +
  \sum_{r=1}^{k-1} A_{i_r i_{r+1}}
  +
  \frac{y_{i_k}}{2}
  -
  \sum_{r=1}^k A_{i_r i_r},
\]
where
\[
  A_{ij}
  =
  \frac{x_j y_i}{2(x_j+y_i)}.
\]

\subsection*{The shadow of a single point}

For a single point $p_i=(x_i,y_i)$, the shadow is
\[
  \tau_i(w)=\max\{w x_i,(1-w)y_i\}.
\]
It consists of two line segments. The branch
\[
  (1-w)y_i
\]
is decreasing in $w$, while the branch
\[
  w x_i
\]
is increasing in $w$. Their intersection is determined by
\[
  w x_i=(1-w)y_i.
\]
Solving gives
\[
  w x_i=y_i-wy_i,
\]
hence
\[
  w(x_i+y_i)=y_i,
\]
and therefore
\[
  c_i=\frac{y_i}{x_i+y_i}.
\]
At this weight,
\[
  c_i x_i
  =
  (1-c_i)y_i
  =
  \frac{x_i y_i}{x_i+y_i}.
\]
The value $c_i$ is the kink position of the shadow $\tau_i$.

The integral of this single shadow is
\[
  \int_0^1 \tau_i(w)\,dw
  =
  \int_0^{c_i} (1-w)y_i\,dw
  +
  \int_{c_i}^1 w x_i\,dw.
\]
The two parts are
\[
  \int_0^{c_i} (1-w)y_i\,dw
  =
  y_i\left(c_i-\frac{c_i^2}{2}\right),
\]
and
\[
  \int_{c_i}^1 w x_i\,dw
  =
  \frac{x_i}{2}(1-c_i^2).
\]
Substituting $c_i=y_i/(x_i+y_i)$ and simplifying gives
\[
  \int_0^1 \tau_i(w)\,dw
  =
  \frac{x_i}{2}
  +
  \frac{y_i}{2}
  -
  \frac{x_i y_i}{2(x_i+y_i)}.
\]
It is therefore natural to define
\[
  A_{ii}
  =
  \frac{x_i y_i}{2(x_i+y_i)}.
\]
Then the one-point case reads
\[
  R_2(\{i\})
  =
  \frac{x_i}{2}
  +
  \frac{y_i}{2}
  -
  A_{ii}.
\]
This calculation explains the diagonal correction term. It is the amount by which the area under the shadow differs from the sum of the two elementary boundary triangles.

\subsection*{The switch between two consecutive selected points}

Now consider two selected points $i<j$. Since the archive is sorted nondominated,
\[
  x_i<x_j,
  \qquad
  y_i>y_j.
\]
Thus point $i$ is better in the first objective, while point $j$ is better in the second objective. On the weight interval this means that point $j$ is favoured near $w=0$, and point $i$ is favoured near $w=1$.

The switch of the lower envelope between these two consecutive selected points is determined by the intersection of the decreasing branch of point $i$ and the increasing branch of point $j$:
\[
  (1-w)y_i=w x_j.
\]
This is the crossing relevant for the lower envelope between consecutive selected points. Solving gives
\[
  w_{ij}
  =
  \frac{y_i}{x_j+y_i}.
\]
At this switch,
\[
  w_{ij}x_j
  =
  (1-w_{ij})y_i
  =
  \frac{x_jy_i}{x_j+y_i}.
\]
We define the corresponding corner term by
\[
  A_{ij}
  =
  \frac{x_jy_i}{2(x_j+y_i)}.
\]
This is one half of the switch height. These terms are the adjacent selected-neighbor transition terms in the final expression.

\subsection*{The active interval of a selected point}

Let
\[
  S=(i_1<i_2<\cdots<i_k)
\]
be the selected chain. Define the switch weights between consecutive selected points by
\[
  \alpha_r
  =
  w_{i_r i_{r+1}}
  =
  \frac{y_{i_r}}{x_{i_{r+1}}+y_{i_r}},
  \qquad r=1,\ldots,k-1.
\]
For notational convenience, set
\[
  \alpha_0=1,
  \qquad
  \alpha_k=0.
\]
Then the selected point $i_r$ is active on the lower envelope over the interval
\[
  [\alpha_r,\alpha_{r-1}].
\]
This convention is a consequence of the orientation of the weight interval: small $w$ emphasizes the second objective, while large $w$ emphasizes the first objective.

The kink position of the active shadow is
\[
  c_{i_r}
  =
  \frac{y_{i_r}}{x_{i_r}+y_{i_r}}.
\]
This kink lies inside the active interval. Indeed, since $x_{i_r}<x_{i_{r+1}}$,
\[
  \alpha_r
  =
  \frac{y_{i_r}}{x_{i_{r+1}}+y_{i_r}}
  <
  \frac{y_{i_r}}{x_{i_r}+y_{i_r}}
  =
  c_{i_r}.
\]
Similarly, since $y_{i_{r-1}}>y_{i_r}$,
\[
  c_{i_r}
  =
  \frac{y_{i_r}}{x_{i_r}+y_{i_r}}
  <
  \frac{y_{i_{r-1}}}{x_{i_r}+y_{i_{r-1}}}
  =
  \alpha_{r-1}.
\]
Thus
\[
  \alpha_r<c_{i_r}<\alpha_{r-1},
\]
with the evident boundary modifications. Hence the integral over the active interval of $i_r$ splits at the kink of $\tau_{i_r}$.

\subsection*{Integral over one active interval}

Fix $r$, and abbreviate
\[
  x=x_{i_r},
  \qquad
  y=y_{i_r},
  \qquad
  a=\alpha_r,
  \qquad
  b=\alpha_{r-1},
  \qquad
  c=\frac{y}{x+y}.
\]
Since $a\le c\le b$, the contribution of $p_{i_r}$ to the lower-envelope integral is
\[
  \int_a^b \tau_{i_r}(w)\,dw
  =
  \int_a^c (1-w)y\,dw
  +
  \int_c^b wx\,dw.
\]
Evaluating the two elementary integrals gives
\[
  \int_a^b \tau_{i_r}(w)\,dw
  =
  y\left(c-\frac{c^2}{2}-a+\frac{a^2}{2}\right)
  +
  x\left(\frac{b^2}{2}-\frac{c^2}{2}\right).
\]
Collecting the terms involving $c$, we obtain
\[
  yc-\frac{(x+y)c^2}{2}.
\]
Since $c=y/(x+y)$, this becomes
\[
  \frac{y^2}{x+y}
  -
  \frac{y^2}{2(x+y)}
  =
  \frac{y^2}{2(x+y)}.
\]
Therefore
\[
  \int_a^b \tau_{i_r}(w)\,dw
  =
  \frac{x b^2}{2}
  -
  ya
  +
  \frac{y a^2}{2}
  +
  \frac{y^2}{2(x+y)}.
\]
This expression is more useful when written in a symmetric boundary form. Observe that
\[
  -ya+\frac{y a^2}{2}+\frac{y}{2}
  =
  \frac{y(1-a)^2}{2},
\]
and
\[
  \frac{y^2}{2(x+y)}-\frac{y}{2}
  =
  -\frac{xy}{2(x+y)}.
\]
Hence
\[
  \int_a^b \tau_{i_r}(w)\,dw
  =
  \frac{x b^2}{2}
  +
  \frac{y(1-a)^2}{2}
  -
  \frac{xy}{2(x+y)}.
\]
Returning to the original notation, this gives
\[
  \int_{\alpha_r}^{\alpha_{r-1}} \tau_{i_r}(w)\,dw
  =
  \frac{x_{i_r}\alpha_{r-1}^2}{2}
  +
  \frac{y_{i_r}(1-\alpha_r)^2}{2}
  -
  A_{i_r i_r}.
\]

\subsection*{Summation and telescoping}

We now sum the active-interval contributions over all selected points:
\[
  R_2(i_1,\ldots,i_k)
  =
  \sum_{r=1}^k
  \left[
    \frac{x_{i_r}\alpha_{r-1}^2}{2}
    +
    \frac{y_{i_r}(1-\alpha_r)^2}{2}
    -
    A_{i_r i_r}
  \right].
\]
The diagonal corrections are already in the desired form:
\[
  -\sum_{r=1}^k A_{i_r i_r}.
\]

It remains to simplify the boundary terms containing the $\alpha_r$. The first boundary value is $\alpha_0=1$, and hence
\[
  \frac{x_{i_1}\alpha_0^2}{2}
  =
  \frac{x_{i_1}}{2}.
\]
The last boundary value is $\alpha_k=0$, and hence
\[
  \frac{y_{i_k}(1-\alpha_k)^2}{2}
  =
  \frac{y_{i_k}}{2}.
\]

Now consider an interior switch between $i_r$ and $i_{r+1}$. The two contributions meeting at this switch are
\[
  \frac{y_{i_r}(1-\alpha_r)^2}{2}
  \quad\text{and}\quad
  \frac{x_{i_{r+1}}\alpha_r^2}{2}.
\]
Using
\[
  \alpha_r
  =
  \frac{y_{i_r}}{x_{i_{r+1}}+y_{i_r}},
  \qquad
  1-\alpha_r
  =
  \frac{x_{i_{r+1}}}{x_{i_{r+1}}+y_{i_r}},
\]
we obtain
\[
  \frac{y_{i_r}(1-\alpha_r)^2}{2}
  +
  \frac{x_{i_{r+1}}\alpha_r^2}{2}
  =
  \frac{
    y_{i_r}x_{i_{r+1}}^2
    +
    x_{i_{r+1}}y_{i_r}^2
  }{
    2(x_{i_{r+1}}+y_{i_r})^2
  }.
\]
Factoring the numerator gives
\[
  \frac{
    x_{i_{r+1}}y_{i_r}
    (x_{i_{r+1}}+y_{i_r})
  }{
    2(x_{i_{r+1}}+y_{i_r})^2
  }
  =
  \frac{x_{i_{r+1}}y_{i_r}}
       {2(x_{i_{r+1}}+y_{i_r})}.
\]
By definition, this is precisely
\[
  A_{i_r i_{r+1}}.
\]

Collecting the two outer boundary terms, the interior switch terms, and the diagonal corrections, we arrive at
\[
  R_2(i_1,\ldots,i_k)
  =
  \frac{x_{i_1}}{2}
  +
  \sum_{r=1}^{k-1} A_{i_r i_{r+1}}
  +
  \frac{y_{i_k}}{2}
  -
  \sum_{r=1}^k A_{i_r i_r}.
\]

This proves the exact adjacent-neighbor decomposition used in the main text. The continuous integral over the lower envelope of the Tchebycheff shadows $\tau_i(w)$ depends, for a sorted selected chain, only on the two boundary points, the selected points through their diagonal corrections, and the adjacent selected pairs. This locality is the structural reason why the Bellman recurrence is exact.

\section{A gentle introduction to matrix search}
\label{app:matrix-search-intro}
Matrix search is a family of algorithms for finding all row minima or all column minima of a matrix faster than by inspecting every entry. The simplest way to view the technique is not as a classical linear algebra method but as an order-exploitation method: if the matrix has enough structure, many entries can be certified irrelevant without being evaluated. The classical reference is the matrix-search algorithm of Aggarwal, Klawe, Moran, Shor, and Wilber, usually referred to by their initials as SMAWK \cite{aggarwal1987smawk}. A broad survey of Monge properties in optimization is given by Burkard, Klinz, and Rudolf \cite{burkard1996monge}.

\subsection{From Monge matrices to monotone minima}

A matrix $C$ is Monge if, for every $i<i'$ and $j<j'$,
\begin{equation}
  C_{ij}+C_{i'j'}\le C_{ij'}+C_{i'j}.
  \label{eq:appendix-monge}
\end{equation}
This inequality says that the north-west plus south-east diagonal of every feasible $2\times2$ submatrix is no larger than the other diagonal. For minimization, this has an important consequence. Suppose row $i'$ is no worse than row $i$ in column $j$, i.e.
\[
  C_{i'j}\le C_{ij}.
\]
Then \eqref{eq:appendix-monge} implies
\[
  C_{i'j'}-C_{ij'}\le C_{i'j}-C_{ij}\le0,
\]
so row $i'$ is still no worse than row $i$ in every later column $j'$. Thus, as one scans columns from left to right, the row index of the uppermost or leftmost column minimum cannot move upwards. This is the column-minimum form of total monotonicity. Rectangular SMAWK exploits this property recursively; the present paper uses a more specialized envelope form because the dynamic-programming matrix is only partially defined.

The distinction between Monge and totally monotone is useful. Monge is a strong $2\times2$ inequality. Total monotonicity is the weaker order condition needed by matrix search: minima remain ordered not only in the whole matrix but also in submatrices. Every Monge matrix is totally monotone, but not every totally monotone matrix is Monge.

\subsection{A small rectangular example}

Consider the matrix $C_{ij}=(i-j)^2$, with rows $i=1,\ldots,4$ and columns $j=1,\ldots,6$:
\[
\begin{array}{c|rrrrrr}
 & c_1 & c_2 & c_3 & c_4 & c_5 & c_6\\
\hline
r_1 & \mathbf{0} & 1 & 4 & 9 & 16 & 25\\
r_2 & 1 & \mathbf{0} & 1 & 4 & 9 & 16\\
r_3 & 4 & 1 & \mathbf{0} & 1 & 4 & 9\\
r_4 & 9 & 4 & 1 & \mathbf{0} & \mathbf{1} & \mathbf{4}
\end{array}
\]
The column-minimizing row indices are
\[
  1,2,3,4,4,4.
\]
They move only downward as the column index increases. For instance, row $3$ beats row $2$ at column $4$, since $1<4$. From that column onward, row $2$ can never again become better than row $3$. This is exactly the kind of dominance that matrix search uses to discard candidates. The algorithm does not need to evaluate all entries in order to know where minima can still occur; it only has to preserve a small ordered set of plausible rows.

The example also illustrates the $2\times2$ Monge inequality. Using rows $2,4$ and columns $3,5$, one obtains
\[
  C_{2,3}+C_{4,5}=1+1=2
  \le
  C_{2,5}+C_{4,3}=9+1=10.
\]
This inequality is the algebraic reason that a later row, once it has become better than an earlier row, remains better in later columns.

\subsection{Why the dynamic-programming matrix is a staircase}

The matrices in the Bellman recurrence are not full rectangles. In layer $r$, column $j$ asks for a predecessor $i<j$, and enough points must have been selected before $i$. Hence the feasible entries are
\[
  \Omega_r=\{(i,j): r-1\le i<j\le n\}.
\]
As $j$ increases, one new predecessor row becomes feasible. The feasible region is therefore a staircase. A tiny schematic example is
\[
\begin{array}{c|ccccc}
 & j=2 & j=3 & j=4 & j=5 & j=6\\
\hline
 i=1 & * & * & * & * & *\\
 i=2 &   & * & * & * & *\\
 i=3 &   &   & * & * & *\\
 i=4 &   &   &   & * & *
\end{array}
\]
where stars denote feasible entries. Padding the blank entries by $+\infty$ is not the cleanest formalization, because artificial infinite ties may obscure the total-monotonicity argument. The right object is a partially defined staircase Monge matrix: every fully feasible $2\times2$ submatrix satisfies the Monge inequality, and feasible row sets are nested prefixes.

\subsection{The envelope interpretation}

For the integral $R_2$ recurrence, each predecessor row $i$ in layer $r$ defines the function
\[
  f_i(x)=D_{r-1}(i)+\frac{x y_i}{2(x+y_i)}.
\]
The value required in column $j$ is the lower envelope value
\[
  \min_{i<j} f_i(x_j).
\]
Because the Pareto-front approximation is sorted with decreasing $y_i$, two such functions cross at most once. Thus the lower envelope is an ordered list of row intervals. The staircase scan performs three simple operations:
\begin{enumerate}
  \item insert the new row $i=j-1$ when it first becomes feasible;
  \item delete previous rows whose future interval of optimality is destroyed by the new row;
  \item query the current envelope at the next increasing abscissa $x_j$.
\end{enumerate}
Each row is inserted once and deleted at most once. The query pointer over the envelope intervals also moves only forward. This is the same order principle as matrix search, specialized to the analytic transition functions of the integral $R_2$ dynamic program.

\subsection{Relation to SMAWK matrix search}

SMAWK is a general-purpose algorithm for rectangular totally monotone matrices under implicit access to entries. It alternates two ideas: eliminate columns that cannot contain a row minimum, and recursively solve a smaller problem before filling the gaps by local search. The algorithm in Section~\ref{sec:matrix-search} does not need the full generality of rectangular SMAWK. Its matrices have a staircase feasible domain and rows enter in the same order as the columns are scanned. Moreover, the entries are generated by single-crossing functions with explicit crossing points. The lower-envelope sweep is therefore a direct, transparent instance of matrix search for this particular recurrence.

The practical message is that direct Bellman recurrence evaluates all feasible predecessor--endpoint pairs in a layer. Divide-and-conquer dynamic programming uses monotone predecessor indices and reduces the search range recursively. Matrix search goes one step further: it keeps only the rows that can still become optimal for some future endpoint. In the bi-objective integral $R_2$ setting, this reduces one layer to linear time while preserving the same exact Bellman recurrence and the same leftmost tie-breaking convention.

\end{document}